\documentclass[a4paper,USenglish]{lipics-v2019}

\usepackage[utf8]{inputenc}
\usepackage[linesnumbered,vlined]{algorithm2e}
\usepackage{amsmath}
\usepackage{xcolor}
\usepackage{amsthm}
\usepackage{comment}
\usepackage{url}
\usepackage{graphicx}
\usepackage[T1]{fontenc}

\newcommand{\minute}{\mathit{minute}}
\newcommand{\hour}{\mathit{hour}}

\newcommand{\ie}{\emph{i.e.},~}

\newcommand{\lmax}{\ell_{max}}

\newcommand{\sinit}{s_{\mathrm{init}}}
\newcommand{\outputs}{\pi_{\mathrm{out}}}
\newcommand{\inputs}{\pi_{\mathrm{in}}}

\newcommand{\etal}{\textit{et al.}\xspace}
\newenvironment{lemma-repeat}[1]{\begin{trivlist}
\item[\hspace{\labelsep}{\bf\noindent Lemma \ref{#1} }]\em }%
{\end{trivlist}}
\newenvironment{theorem-repeat}[1]{\begin{trivlist}
\item[\hspace{\labelsep}{\bf\noindent Theorem \ref{#1} }]\em }%
{\end{trivlist}}

\newcommand{\remove}[1]{}

\newcommand{\size}[1]{\ensuremath{\left|#1\right|}}
\newcommand{\set}[1]{\left\{ #1 \right\}}
\DeclareMathOperator{\polylog}{polylog}

\title{Smoothed Analysis of Population Protocols} 

\titlerunning{} 

\author{Gregory Schwartzman}{JAIST, Japan  }{greg@jaist.ac.jp}{}{}
\author{Yuichi Sudo}{Hosei University, Japan  }{sudo@hosei.ac.jp}{}{}

\authorrunning{Gregory Schwartzman and Yuichi Sudo} 

\Copyright{Gregory Schwartzman and Yuichi Sudo} 



\category{} 

\relatedversion{} 

\supplement{}

\nolinenumbers 

\hideLIPIcs  

\newcommand{\myparagraph}[1]{\bigskip\noindent{\textbf{#1}}}

\begin{document}
\maketitle
\begin{abstract}
    In this work, we initiate the study of \emph{smoothed analysis} of population protocols. We consider a population protocol model where an adaptive adversary dictates the interactions between agents, but with probability $p$ every such interaction may change into an interaction between two agents chosen uniformly at random. That is, $p$-fraction of the interactions are random, while $(1-p)$-fraction are adversarial.
    The aim of our model is to bridge the gap between a uniformly random scheduler (which is too idealistic) and an adversarial scheduler (which is too strict).
    
    We focus on the fundamental problem of leader election in population protocols. We show that, for a population of size $n$, the leader election problem can be solved in $O(p^{-2}n \log^3 n)$ steps with high probability, using $O((\log^2 n) \cdot (\log (n/p)))$ states per agent, for \emph{all} values of $p\leq 1$. Although our result does not match the best known running time of $O(n \log n)$ for the uniformly random scheduler ($p=1$), we are able to present a \emph{smooth transition} between a running time of $O(n \polylog n)$ for $p=1$ and an infinite running time for the adversarial scheduler ($p=0$), where the problem cannot be solved.
    The key technical contribution of our work is a novel \emph{phase clock} algorithm for our model. This is a key primitive for much-studied fundamental population protocol algorithms (leader election, majority), and we believe it is of independent interest.
\end{abstract}

\section{Introduction}
In the traditional population protocol model \cite{AAD+06}, we have a population of $n$ agents, where every agents is a finite state machine with a small number of states. We refer to the cross product of all of the states of the agents in the population as the \emph{configuration} of the population.
Two agents can interact, whereupon their internal states may change as a \emph{deterministic} function of their current states. While the transition function is deterministic, it need not be symmetrical. That is, the interaction is ordered, where one agent is called an \emph{initiator} and the other is called a \emph{responder}. A standard assumption is that the order of agents upon an interaction is chosen uniformly at random. This is equivalent to having a single random bit that can be used by the transition function.

The sequence of interactions that the population undergoes is called a \emph{schedule}, and is decided by a \emph{scheduler}. The standard scheduler used in this model is the \emph{uniformly random} scheduler, which chooses all interactions uniformly at random. Agent states are mapped to outputs via a problem specific output function.
Generally, a protocol aims to take any legal initial configuration (legal input) and, after a sufficient number of interactions, turn it into one of a desired set of configurations (legal output). After the population reaches a legal output configuration, every following configuration is also a legal output. The running time of the protocol is an upper bound on the number of interactions (steps) required to map any legal input to a legal output. The model assumes agent interactions happen sequentially, but another common term is \emph{parallel time}, which is the running time of a protocol divided by $n$. A formal definition of population protocols is given in Section~\ref{sec: pre}.

In this paper, we focus on the \emph{leader election} problem. In this problem, every agent is initially marked as either a \emph{leader} or a \emph{follower}. We are guaranteed that initially there exists at least one leader in the population. The goal is to design a protocol, such that, after a sufficient number of interactions, the population always converges to a configuration with a unique leader. The leader election problem has received a large amount of attention in the population protocol literature \cite{AAD+06,AAF+08,AG15,AAE+17,AAG18,GS20,GSU19,SudoOIKM20,SNY12,Sud+20,CC19,CC20,SSN+21,YSM20,FJ06,BBB13,BDN+19,SIW12,SEI+20} and thus makes the perfect case study for our model.

\myparagraph{Motivation for our model} Population protocols aim to model the computational power of a population of many weak computational entities. Initially introduced to model animal populations \cite{AAD+06} (a flock of birds, each attached with a sensor), the model has found use in a wide range of fields. For example: wireless sensor networks \cite{PerronVV09,DraiefV12}, molecular computation (e.g. DNA computing) \cite{chen2013programmable,ChenCDS17,cardelli2012cell}.
The assumption of completely uniform interactions in these models is a reasonable \emph{approximation} to the true nature of the interactions. That is, a flock of birds does not interact uniformly at random, the interaction probability of molecules in a fluid can depend on their size and shape, and sensors in a sensor network may experience delays, malfunctions, or even adversarial attacks. The common thread among all of these scenarios is that, while the uniformity assumption is too strong, these systems still contain some amount of randomness. There is a rich literature on designing population protocols for the uniformly random scheduler, and it would be very disheartening if these results do not generalize if we slightly weaken the scheduler.
In this paper we try to model these environments, which are "somewhat noisy", and answer the question: Are current population protocol algorithms robust or fragile?

\myparagraph{Smoothed analysis} To this end we consider \emph{smoothed analysis} of population protocols. Smoothed analysis was first introduced by Spielman and Teng~\cite{SpielmanT09, SpielmanT04}, in an attempt to  explain the fact that some problems admit strong theoretical lower bounds, but in practice are solved on a daily basis.
The explanation smoothed analysis suggests for this gap, is that lower bounds are proved using very specific, pathological instances of a problem, which are highly unlikely to happen in practice. They support this idea by showing that some lower bound instances are extremely \emph{fragile}, i.e., a small random perturbation turns a hard instance into an easy one.
Spielman and Teng applied this idea to the simplex algorithm, and showed that, while requiring an exponential time in the \emph{worst case}, if we apply a small random noise to our instance before executing the simplex algorithm on it, the running time becomes polynomial in expectation. 

While in classical algorithm analysis worst-case analysis currently reigns supreme, the opposite is true regarding population protocols. The vast majority of population protocols assume the uniformly random scheduler. This is due to the fact that under the adversarial scheduler most problems of interest are pathological. This reliance on the uniformly random scheduler leads us to ask the following questions:
Is the assumption regarding a uniformly random scheduler too strong? Will the algorithms developed under this assumption fail in the real world? We use smoothed analysis to show that indeed it is possible to design \emph{robust} algorithms for the much-studied leader election problem in population protocols. That is, an algorithm that can provide convergence guarantees even if only a \emph{tiny} percentage of the interactions is random. In doing so we smoothly bridge the gap between the adversarial scheduler and the uniformly random scheduler. 

\myparagraph{Our model and results} It is easy to see that if \emph{all} of the interactions are chosen adversarially, no problem of interest can be solved. In this paper, we present a model which smoothly bridges the gap between the adversarial scheduler and the uniformly random scheduler. In our model, we have an adaptive adversary which chooses the $(i+1)$-th interaction for the population after the completion of the $i$-th interaction. This choice can be based on \emph{any} past information of the system (interactions, configurations, random bits flipped). With probability $1-p$ the next interaction is the one chosen by the adversary, and with probability $p$ it is an interaction between two agents chosen uniformly at random. We call $p$ the smoothing parameter. In our analysis we allow protocols to access randomization directly as in \cite{CP07,MST18,BGK20}. Specifically, we assume that each time two agents interact, they get one (unbiased) random bit. In the appendix, we show how to extend our results even if we only assume that the order of initiator-responder is random (and no random bit is given). This results in a slight slowdown in our convergence time by a $O(p^{-1}\log n)$ multiplicative factor. If we assume a random bit is flipped for every interaction, then the adversary cannot decide the outcome of the random bit (but may decide the initiator-responder order). While if we assume no random bit is flipped, then the adversary cannot decide the initiator-responder order, and it is taken to be random.

This model is meant to model an environment that is mostly adversarial, but contains a small amount of randomness. We consider the fundamental problem of leader election in this model and show that we can design a protocol for our model which uses $O((\log^2 n) \cdot (\log (n/p)))$ states per agent, and elects a unique leader in $O(p^{-2}n \log^3 n)$ steps with high probability. Although our result does not match the best known running time of $O(n \log n)$, using $O(\log \log n)$ states, for the uniformly random scheduler ($p=1$), we are able to present a \emph{smooth transition} between a running time of $O(n \polylog n)$, using $O(\polylog n)$ states, for $p=1$ and an infinite running time for the adversarial scheduler ($p=0$), where the problem cannot be solved, regardless the number of states.

Furthermore, this shows that \emph{any} amount of noise in the system is sufficient to guarantee that the leader election problem can be solved if we allow for a sufficient number of states per agent. We also note that because the number of states required is $O((\log^2 n) \cdot (\log (n/p)))$, even for an extremely minuscule amount of noise, $p=1/poly(n)$, the leader election problem can be solved by agents using $\polylog (n)$ states.
This is important because we would like our agents to be very simple computational units, so we would like to avoid agents with a super-polylogarithmic number of states.

The key building block in our leader election algorithm is the \emph{phase clock} primitive (see section~\ref{sec: phase clock} for a formal definition). This is a weak synchronization primitive, which is at the heart of many state of the art algorithms for fundamental problems like leader election and majority in population protocols \cite{AAE08, AAG18, GS20, GSU19, BGK20, SudoOIKM20, SEI+20,BEF+20}. The analysis for all current \emph{phase clock} implementations fails for any constant smoothing parameter $p<1$, assuming an $O(\mathrm{polylog}(n))$ number of states.
Roughly speaking, existing phase clocks break when the adversary chooses two agents and repeatedly forces them to interact (a detailed explanation is given in Section~\ref{sec: why existing break}).

We present a novel phase clock design that is robust even when all but a tiny fraction of the interactions are adversarial. Our phase clock relies heavily on the fact that the random bits flipped per interaction are not chosen in an adversarial fashion. To overcome the shortcomings of existing phase clock algorithms, we base our phase clock on a stochastic process whose correctness is \emph{indifferent} to adversarial interactions. 
Finally, we show that using our phase clock in a simplified (and slower) version of the leader election algorithm of \cite{SudoOIKM20} achieves the desired running time. 
Although we provide a complete (and simplified) proof of the leader election protocol with our phase clock for completeness, the original analysis \cite{SudoOIKM20} still goes through unchanged. That is, our phase clock is basically plug-and-play. Thus, we believe this primitive can be used directly for more complex population protocols such as the complete leader election algorithm of \cite{SudoOIKM20}, or the majority algorithm of \cite{AAG18,BEF+20}. However, properly presenting and analyzing these algorithms is beyond the scope of this paper, and we leave it for future work.

\subsection{Related Work}
Smoothed analysis was introduced by Spielman and Teng~\cite{SpielmanT09,SpielmanT04}. Since  then, it has received much attention in sequential algorithm design (see the survey in~\cite{SpielmanT09}). Recently, smoothed analysis has also received some attention in the distributed setting.
The first such application is due to Dinitz \etal~\cite{DFGN18}, who apply it to various well-studied problems in dynamic networks.
Since then, different smoothing models \cite{MeirPS20} and problems \cite{ChatterjeePP20,MollaS20} were considered. To the best of our knowledge, we are the first to consider smoothed analysis of population protocols.

Leader election has been extensively 
studied in the population protocol model. 
The problem was first considered in \cite{AAD+06}, where a simple protocol was presented. In this protocol, all agents are initially leaders, and we have only one transition rule: when two leaders meet, one of them becomes a follower (\ie a non-leader). This simple protocol uses only two states per agent
and elects a unique leader 
in $O(n^2)$ steps in expectation. 
This protocol is time-optimal:
Doty and Soloveichik \cite{DS18} showed that any constant space protocol requires $\Omega(n^2)$ expected steps to elect a unique leader. 
In a breakthrough result, Alistarh and Gelashvili \cite{AG15} 
designed a leader election protocol 
that converges in $O(n\log^3n)$ expected steps 
and uses $O(\log^3 n)$ states per agent. 
Thereafter, a number of papers have been devoted to 
fast leader election \cite{AAG18,GS20,GSU19,SudoOIKM20,BGK20}.
G{\k{a}}sieniec, Staehowiak, and Uznanski \cite{GSU19} 
gave an algorithm that converges in 
$O(n \log n \log \log n)$ expected steps 
and uses a surprisingly small number of states:
only $O(\log \log n)$ states per agent.
This is space-optimal because it is known that
every leader election protocol with a $O(n^2/\mathrm{polylog}(n))$ convergence time
requires $\Omega(\log \log n)$ states \cite{AAE+17}. 
Sudo \etal~\cite{SudoOIKM20} gave
a protocol that elects a unique leader within 
$O(n\log n)$ expected steps 
and uses $O(\log n)$ states per agent. 
This is time-optimal because
any leader election protocol 
requires $\Omega(n\log n)$ expected steps 
even if it uses an arbitrarily large number of states 
and the agents know the exact size of the population \cite{SM20}.
These two protocols were the state-of-the-art 
until recently, when Berenbrink \etal~\cite{BGK20}
gave a time and space optimal protocol.
In all of the above literature, the stabilization time (\ie the number of steps it takes to elect a unique leader)
is evaluated under the uniformly random scheduler.

Self-stabilizing leader election has also been well studied 
\cite{AAF+08,SNY12,Sud+20,CC19,CC20,SSN+21,YSM20,FJ06,BBB13,BDN+19,SIW12,SEI+20}. In the self-stabilizing setting, we do not assume that all agents are initialized at the beginning of an execution.
That is, we must guarantee that a single leader is elected eventually and maintained thereafter
even if the population begins an execution from an \emph{arbitrary} configuration. 
Typically, the population must create a new leader if there is no leader initially,
while the population must decrease the number of leaders to one 
if there are two or more leaders initially.
Unfortunately, the self-stabilizing leader election cannot be solved
in the standard model \cite{AAF+08}.
Thus, this problem has been considered 
(i) by assuming that the agents have global knowledge such as the exact number of agents
\cite{SIW12,BDN+19,SSN+21},
(ii) by assuming the existence of oracles \cite{FJ06,BBB13},
(iii) by slightly relaxing the requirement of self-stabilization \cite{SNY12,Sud+20,SEI+20},
or (iv) by assuming a specific topology of the population such as rings \cite{AAF+08,CC19,CC20,YSM20}.

Several papers on population protocols assume the globally fair scheduler \cite{AAD+06,AAF+08,FJ06,BBB13,CC19}. 
Intuitively, this scheduler cannot avoid a possible step forever. 
Formally, the scheduler guarantees that 
in an infinite execution, a configuration appears infinitely often
if it is reachable from a configuration that appears infinitely often.
Assuming the fairness condition is very helpful in designing protocols that solve some problem \emph{eventually}, however, it is not helpful in bounding the stabilization time. 
Thus, the uniformly random scheduler is often assumed to evaluate the time complexities of protocols, as mentioned above. Actually, the uniformly random scheduler is a special case of the globally fair scheduler.

Several papers considered population protocols with some form of noise. In \cite{XuBBKN20} a random scheduler with non-uniform interaction probabilities is proposed, and the problem of data collection is analyzed for this model. While their model generalizes the standard random scheduler, it still does not allow adversarial interactions, and thus is quite different from our model. Sadano \etal \cite{SSK+19} introduced and considered a stronger model than the original population protocols under the uniformly random scheduler. In their model, agents can control their moving speeds. Faster agents have a higher probability to be selected by the scheduler at each step. They show that some protocols have a much smaller stabilization time by changing the speeds of agents. 

Similarly to us, the authors of \cite{BeauquierBBG15} also try to answer the question of whether population protocols can function under imperfect randomness. They take a very general approach, which is somewhat different than ours. 
The main difference is that the randomness of a schedule (a sequence of interactions) is measured as its \emph{Kolmogorov complexity} (the size of the shortest Turing machine which outputs the schedule). Intuitively, the schedule is random if its Kolmogorov complexity is (almost) equal to the length of the schedule. They parameterize the "randomness" of the schedule by a parameter $T$, where $T=1$ means that the schedule is completely random, while the randomness decreases as the value of $T$ goes to 0. An \emph{adversary} with parameter $T$ is a scheduler that only generates schedules with parameter $T$.

They show that any problem which can be solved for $T=1$ can also be solved for $T<1$ (imperfect randomness). They also consider the leader election problem, and give upper and lower bounds for the value of $T$ required to solve the problem. Their bounds are not explicit, but are presented as a function of the largest root of a certain polynomial. Apart from a different notion of randomness, their work differs from ours in that it assumes an oblivious adversary (while we consider an adaptive adversary). This makes a direct comparison between our results and those of \cite{BeauquierBBG15} somewhat tricky. It might be said that we take a somewhat more pragmatic approach, showing a very natural augmentation to the popular random scheduler. This allows us to \emph{explicitly} express the running time and number of states for all values of $p$ (showing that for reasonable values, the performance is very close to that of the random scheduler), while in \cite{BeauquierBBG15} only existence results are presented.

\section{Preliminaries}
\label{sec: pre}

\subsection{Population Protocols}

A \emph{population} is
a network consisting of {\em agents}.
We denote the set of all agents by $V$ and let $n = |V|$.
We assume that a population is a complete graph,
thus every pair of agents $(u,v)$ can interact,
where $u$ serves as the \emph{initiator}
and $v$ serves as the \emph{responder} of the interaction.
Throughout this paper, we use the phrase
``with high probability'' to denote a
probability of $1-O(n^{-\alpha})$ for an arbitrarily large constant $\alpha$.

A \emph{protocol} $P(Q,T,X,Y,\inputs,\outputs)$ consists of 
a finite set $Q$ of states,
a transition function
$T:  Q \times Q \times \{0,1\} \to Q \times Q$,
a finite set $X$ of input symbols,
a finite set $Y$ of output symbols,
an input function $\inputs: X \to Q$,
and an output function $\outputs : Q \to Y$.
The agents are given (possibly different) inputs $x \in X$.
The input function $\inputs$ determines their initial states $\inputs(x)$.
When two agents interact,
$T$ determines their next states
according to their current states and one bit.
The \emph{output} of an agent is determined by $\outputs$:
the output of an agent in state $q$ is $\outputs(q)$.

A \emph{configuration} is a mapping $C : V \to Q$ that specifies
the states of all the agents.
We say that a configuration $C$ changes to $C'$ by the interaction
$e = (u,v)$
and a bit $b$,
denoted by $C \stackrel{(e,b)}{\to} C'$,
if
$(C'(u),C'(v))=T(C(u),C(v),b)$
and $C'(w) = C(w)$
for all $w \in V \setminus \{u,v\}$.

Thus, given a configuration $C$, a sequence of interactions (or ordered pairs of agents) $\{\gamma_i\}_{i=0}^\infty$,
and a sequence of bits $\{b_i\}_{i=0}^\infty$,
the \emph{execution} starting from $C$ under $\{\gamma_i\}_{i=0}^\infty$ and $\{b_i\}_{i=0}^\infty$
is defined as the sequence of configurations $\{C_i\}_{i=0}^\infty$ such that $C_i \stackrel{(\gamma_i,b_i)}{\to} C_{i+1}$.
A sequence of interactions is called a \emph{scheule} and will be explained in detailed in the next subsection.
We assume that each $b_i \in \{0,1\}$ is a random variable such that $\Pr[b_i = 1] = 1/2$
and these random bits $b_0,b_1,\dots$ are independent of each other.
That is, upon each interaction the two agents have access to a bit of randomness to decide their new states.
In the appendix, we show how to extend our results for the more standard population protocol model where only the order of initiator-responder is random, and no additional randomness is available.


\subsection{Schedulers}
A schedule $\gamma=\set{\gamma_i}_{i=0}^{\infty}=\set{(u_i,v_i)}_{i=0}^{\infty}$ is a sequence of ordered pairs which determines the interactions the population of agents undergoes. Note that although $\gamma$ is ordered, we use a set notation for simplicity. The schedule is determined by a \emph{scheduler}. In the classical population protocol model, a uniformly random scheduler is used. That is, every pair in $\gamma$ is chosen uniformly at random. Let us denote this scheduler by $\Gamma_{u}$. One can also consider an \emph{adversarial} scheduler. Such a scheduler creates $\gamma$ in an adversarial fashion. Let us denote this scheduler by $\Gamma_a$. We would like to note that while the sequence of interactions is chosen adversarially by $\Gamma_a$, it does not determine the coin flips observed by the agents.
This type of scheduler can either be \emph{adaptive} or \emph{oblivious} (non-adaptive). In both cases the adversary has complete knowledge of the initial state of the population and the algorithm executed by the agents. However, for the oblivious case the sequence of interactions, $\gamma$, must be chosen \emph{before} the execution of the protocol, while for the adaptive case the interaction $\gamma_i$ is chosen by the adversary after the execution of the $(i-1)$-th step, with full knowledge of the current state of the population. The difference between an oblivious and an adaptive adversary can also be stated in term of knowledge of the randomness in the population. An adaptive adversary has full knowledge regarding the population, including the random coins used in the past. While the oblivious adversary does not have access to the randomness of the system.  

\myparagraph{Our model}
We consider a smoothed scheduler $\Gamma_s$ which is a combination of $\Gamma_u$ and $\Gamma_a$. Specifically, let $\gamma^a=\set{\gamma^a_i}_{i=0}^{\infty},\gamma^u=\set{\gamma^u_i}_{i=0}^{\infty}$ be the schedules chosen by $\Gamma_a,\Gamma_u$. We define the smoothed schedule $\gamma_s=\set{\gamma^s_i}_{i=0}^{\infty}$ of $\Gamma_s$ as $\gamma^s_i = \gamma^u_i$ with probability $p$ and $\gamma^s_i = \gamma^a_i$ with probability $1-p$, where $p\in [0,1]$ is the \emph{smoothing parameter}. We note that if $\Gamma_a$ is adaptive, then so is $\Gamma_s$. For the rest of the paper we focus on an \emph{adaptive} adversary.

In this paper, we assume that a rough knowledge
of an upper bound of $n, p^{-1}$ is available.
Specifically, we assume that all agents know two common values $n',p'$, such that $n\leq n' = O(n), p \geq p' = \Omega(p)$. Assuming such a rough knowledge about $n$ is standard in the recent population protocol literature \cite{AG15, AAE+17, AAG18, GSU19, MST18, BCER17, BGK20, SNY12, SudoOIKM20, Sud+20, SEI+20, BEF+20}, and we generalize this assumption for $p$. Due to the asymptotic equivalence between $n, p$ and $n', p'$, we only use only $n,p$ in the definition and analysis of our algorithm.

\subsection{Leader election}
The leader election problem requires that 
every agent should output $L$ or $F$ (``leader'' or ``follower'') respectively.
We say that a configuration $C$ of $P$ is output-stable if 
no agent may change its output in an execution of $P$ that starts from $C$,
regardless of the choice of interactions.
Let $\mathcal{S}_P$ be the set of the output-stable configurations
such that, for any configuration $C \in \mathcal{S}_P$,
exactly one agent outputs $L$ (\ie is a leader) in $C$.
This problem does not require inputs for the agents.
Hence, we assume $X = \{x\}$,
thus all the agents begin an execution with a common state $\sinit = \inputs(x) $.
We say that a protocol $P$
is a leader election protocol for a scheduler $\Gamma$, if the execution of the protocol starting from
the configuration where all agents are in state $\sinit$
reaches a configuration in $\mathcal{S}_P$ with probability $1$ with respect to the scheduler $\Gamma$.
We define the stabilization time of the execution as the number of steps until it reaches a configuration in $\mathcal{S}_P$ for the first time.
\subsection{One-way Epidemic}

In the proposed protocol,
we often use the \emph{one-way epidemic} protocol\cite{AAE08}.
This is a population protocol where every agent has two states $\set{0,1}$ and the transition function is given as $(x,y) \rightarrow (x, \max \set{x,y})$. 
All nodes with value 1 are \emph{infected}, while all nodes with value 0 are \emph{susceptible}. Initially, we assume that a single node is infected. We say that the one-way epidemic finishes when all nodes are infected. This is an important primitive for spreading a piece of information among the population.

Angluin et al.~\cite{AAE08} prove that
one-way epidemic finishes within $\Theta(n \log n)$ interactions with high probability for $\Gamma_u$. It is easy to see that the one-way epidemic protocol finishes within $O(p^{-1}n \log n)$ steps for $\Gamma_s$ with high probability. This is because within $O(p^{-1}n \log n)$ steps of $\Gamma_s$ there must exist $\Omega(n \log n)$ \emph{random} interactions with high probability. Due to the nature of the one-way epidemic protocol, we can just ignore all adversarial interactions, and the original analysis goes through.

\subsection{Martingale concentration bounds}
In our analysis we often encounter the following scenario: We have a series of \emph{dependent} binary random variables $\set{X_i}_{i\geq0}$ such that $\forall i, E[X_i \mid X_0,...,X_{i-1}] \geq q$, for some constant $q$.\footnote{Note that this condition is equivalent to $\forall i, Pr[X_i=1 \mid X_0=s_0,...,X_{i-1}=s_{i-1}] \geq q$ for any binary string $s$ of length $i-1$.} 
And we would like to bound the probability $Pr[\sum_{i=0}^{\lceil \alpha q^{-1}t\rceil}X
_i \leq  t]$ for some constant $\alpha$. Note that if the variables were independent, we could have simply used a Chernoff type bound. As this is not the case, we use martingales for our analysis.

 We say that a sequence of random variables, $\set{Y_i}_{i=0}$, is a sub-martingale with respect to another sequence of random variables $\set{X_i}_{i\geq0}$ if it holds that $\forall i, E[Y_i \mid X_0,...,X_{i-1}] \geq Y_{i-1}$. The following concentration equality holds for sub-martingales:

\begin{theorem}[Azuma]
Suppose that $\set{Y_i}_{i\geq0}$ is a sub-martingale with respect to $\set{X_i}_{i\geq0}$, and that $|Y_i - Y_{i-1}|\leq c_i$. Then for all positive integers $k$ and positive reals $\epsilon$ it holds that:
$$Pr[Y_k - Y_0 < -\epsilon] \leq e^{\frac{-\epsilon^2}{2\sum^k_{j=1} c_j}}$$

\end{theorem}

Let us consider the sequence $\set{X_i}_{i\geq0}$ from before. Recall that $\forall i, E[X_i \mid X_0,...,X_{i-1}] \geq q$.
Without loss of generality assume that $X_0=0$.

Let us define $Y_i = \sum_{j=0}^i X_j - q\cdot i$. Note that $Y_0 = 0$. Let us show that $\set{Y_i}_{i\geq0}$ is a a submartingale with respect to $\set{X_i}_{i\geq0}$. It holds that:
\begin{align*}
    &E[Y_{i+1} \mid X_0,...,X_i] = E[X_{i+1}-q + Y_{i} \mid X_0,...,X_i]
    \\&=E[X_{i+1} \mid X_0,...,X_i] - q + Y_{i} > Y_{i}
\end{align*}
Where in the transitions we used the fact that the variables $X_0,..,X_i$ completely determine $Y_i$, thus $E[Y_{i} \mid X_0,...,X_i]=Y_i$, and the fact that $E[X_i \mid X_0,...,X_{i-1}] \geq q$.
Next we apply Azuma's inequality for $Y_{\lceil 2 q^{-1} t \rceil}$ by setting $\epsilon= t$, noting that $\forall i, |Y_{i+1}-Y_i| \leq 1$.
\begin{align*}
\Pr\left[\sum^{\lceil 2 q^{-1} t \rceil}_{i=0}X_i < t\right] 
\leq \Pr\left[\sum^{\lceil 2 q^{-1} t \rceil}_{i=0}X_i -2t \leq  - t\right]
\le \Pr[Y_{\lceil 2 q^{-1} t \rceil} < -t]
\le e^{-t^2/\lceil 2q^{-1} t \rceil} = e^{-\Theta(t)}.
\end{align*}

We state the following theorem:

\begin{theorem}
\label{thm: modified chernoff}
Let $\set{X_i}_{i\geq0}$ be a series of binary random variables such that $\forall i, E[X_i \mid X_0,...,X_{i-1}] \geq q$ for some constant $q$. Then it holds that for every positive integer $t$:
\begin{align*}
\Pr\left[\sum^{\lceil 2 q^{-1} t \rceil}_{i=0}X_i < t\right] 
\le  e^{-\Theta(t)}
\end{align*}
\end{theorem}

Specifically, when we set $t=\Theta(\log n)$ with a sufficiently large constant we get a high probability bound.


\section{Phase clock implementation for $\Gamma_s$}
\label{sec: phase clock}
A phase clock is a weak synchronization primitive used in population protocols. In a phase clock we would like all of the agents to have a variable, let's call it \emph{hour}, with the following properties:
\begin{enumerate}
    \item All agents simultaneously spend $\Omega(f(n))$ steps in the same hour.
    \item For every agent an hour lasts $O(g(n))$,    
\end{enumerate}
Where the above holds with high probability for every value of hour.
Ideally, we desire $f(n) = g(n)$. 

We borrow some notation from \cite{berenbrink2020time}, and define the above more formally. A \emph{round} is a period of time during which all agents have the same $hour$ value. Denote by $R_s(i), R_e(i)$ the start and end of round $i$. Formally, $R_s(i)$ is the interaction at which the last agent reaches hour $i$, while $R_e(i)$ is the interaction during which the first agent reaches hour $i+1$. We define the length of round $i$ as $L(i)=\max \set{0, R_e(i) - R_s(i)}$. Note that it may be the case that $R_e(i) \leq R_s(i)$, thus a $\max$ is needed in the definition. Finally, during these $L(i)$ interactions, all agents have $hour=i$. Next we define the stretch of round $i$ as $S(i)=R_e(i) - R_e(i-1)$. This is the amount of time since $hour=i$ is reached for the first time until $hour=i+1$ is reached for the first time. Note that $L(i) \leq S(i)$ always holds. For a visual representation, we refer the reader to Figure~\ref{fig: rounds}.
\begin{figure}
\centering
  \includegraphics[scale=0.3]{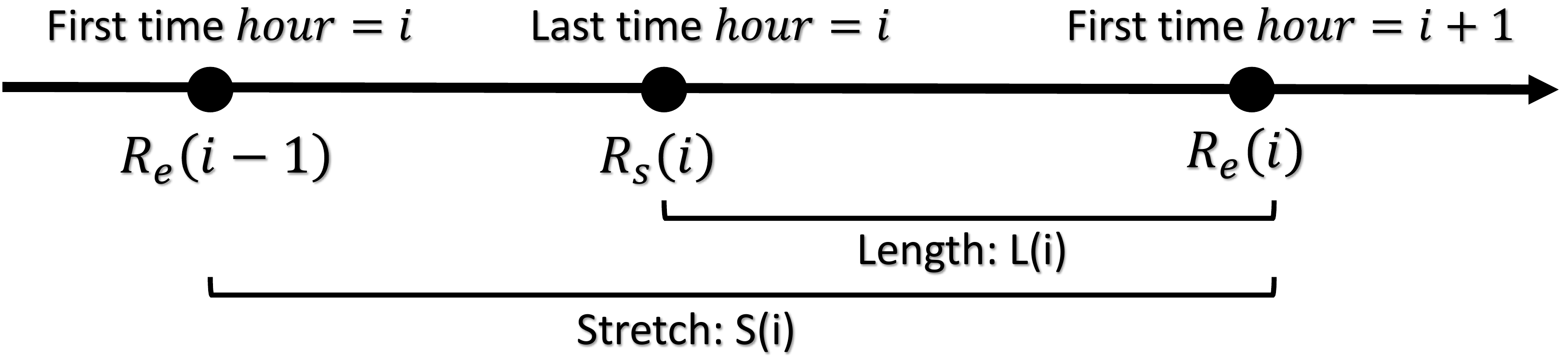}
  \caption{A visual representation of the length and stretch of a round.}
  \label{fig: rounds}
\end{figure}
Using the above notation, we define a phase clock.

\begin{definition}
\label{thm: phase clock def}
We say that an algorithm is  
a phase clock with parameters $f(n)$ and $g(n)$ (or a $(f(n),g(n))$-phase clock)
if it has the following guarantees with high probability for any $i \ge 0$:
\begin{enumerate}
    \item $L(i) \geq d_1 f(n)$
    \item $S(i) \leq d_2 g(n)$
\end{enumerate}
Where $d_1$ and $d_2$ are adjustable constants (taken to be sufficiently large). 
When $f(n) = g(n)$ we simply write a $f(n)$-phase clock.
\end{definition}
We note that all current phase clock algorithms require the uniformly random scheduler, and do not extend to $\Gamma_s$, as we will see in the next subsection.

\subsection{Why existing algorithms fail}
\label{sec: why existing break}
In this subsection, we show why the existing phase clock algorithms fail in our model.

There are three kinds of phase clock algorithms in the field of population protocols:
a phase clock with a unique leader \cite{AAE08},
a phase clock with a junta \cite{GS20,GSU19,BGK20,BEF+20},
and a leaderless phase clock \cite{AAG18,SudoOIKM20}. 
The first kind is essentially a special case of the second kind. 
The second kind is a $\log n$-phase clock that uses only a constant number of states.
However, we require the assumption that there is a set $J \subset V$ of agents marked as members of a \emph{junta}, such that $|J| = O(n^{1-\epsilon})$, where $\epsilon$ is a constant. 
The third kind is a $\log n$-phase clock
that uses $O(\log n)$ states but does not require the existence of a junta. 

In our notation, the second algorithm can be written as follows:
\begin{itemize}
    \item Each agent has a variable $\minute \in \mathbb{N}$. 
    \item Each agent outputs $hour = \lfloor minute / M \rfloor$, where $M$ is a (sufficiently large) constant. 
    \item Suppose that an initiator $u$ and a responder $v$ interact.
    The initiator $u$ sets its $minute$ to \sloppy{$\max(u.minute, v.minute+1)$} if $u$ is in the junta;
    otherwise to $\max(u.minute, v.minute)$.
\end{itemize}
By the definition of the algorithm, 
only an agent in the junta can increase $\max_{v \in V} v.minute$. 
This fact and the sublinear size of the junta guarantees that
the length of each round is $\Omega(n \log n)$ with high probability.
However, this guarantee depends on the uniformly random nature of the scheduler.
In our model, every interaction is chosen adversarially with probability $1-p$. 
The adversary can force two agents in the junta, say $u$ and $v$, to interact so frequently that
every round finishes within $O(1/(1-p))$ steps.
Thus, unless $p=1-O(1/n)$, \ie unless the adversary can only choose an extremely small fraction of the interactions, the adversary can always force each round to finish in $o(n)$ steps, \ie $o(1)$ parallel time. In particular, if $p=1-\Omega(1)$,
the adversary can always force each round to finish in a constant number of steps.

The third algorithm (the leaderless phase clock) 
can be written as follows\footnotemark:
\footnotetext{
The implementation of this phase clock slightly differs between \cite{AAG18} and \cite{SudoOIKM20}. Here we describe the implementation presented in \cite{SudoOIKM20}.
}
\begin{itemize}
    \item Each agent has variables $hour \in \mathbb{N}$
    and $minute \in \{0,1,\dots,M\}$, where $M = \Theta(\log n)$ with a sufficiently large hidden constant.
    \item Suppose that an initiator $u$ and a responder $v$ interact.
    The initiator $u$ updates its $hour$ and $minute$ as follows:
    $$(u.hour, u.minute) \gets \begin{cases}
        (v.hour, 0) & \text{if}\ u.hour < v.hour\\
        (u.hour + 1, 0)& \text{else if}\ u.minute = M\\
        (u.hour,u.minute+1) & \text{otherwise}.
    \end{cases}$$
\end{itemize}
In this phase clock, 
an agent resets its $minute$ to zero
each time it increases its $hour$.
Once an agent resets its $minute$ to zero,
it must have no less than $M$ interactions,
or interact with an agent whose $hour$ is larger than its $hour$,
before it increases its $hour$.
Thus, one can easily observe that 
the length of each round is $\Omega(n \log n)$
under the uniformly random scheduler.
However, in our model, this does not hold.
The adversary can pick two agents and force them to interact frequently, so that every round finishes within $O((\log n)/(1-p))$ steps, \ie $O((\log n)/(n(1-p)))$ parallel time.

\subsection{Our algorithm}
We present a $((np^{-1} \log^2 n), (np^{-2} \log^2 n))$-phase clock using $O((\log n) \cdot \log (n/p))$ states per agent,
where $p$ is the smoothing parameter for $\Gamma_s$. 

In our algorithm each agent $v$ has three states: $second,minute,hour$. Where $hour$ is the output variable. The domain of the variables is: $second \in \set{0,...,S}, minute \in \set{0,...,M}$ where $S = \log (n/p) + \log \log n + c, c=O(1)$ and $M=\Theta(\log n)$. For simplicity of notation and without loss of generality, we assume that $1/p, \log n, M, S, c$ are all integers.
While the domain of $hour$ is unbounded in our algorithm, it can easily be taken to be bounded (By using a simple modulo operation \cite{berenbrink2020time}, or by stopping the counter once it reaches some upper limit \cite{GS20,GSU19}). 
All variables are initialized to 0. For each interaction, we apply Algorithm~\ref{alg: phase clock}, where $u$ is the initiator and $v$ is the responder. Roughly speaking, the $second$ variable follows the following random walk pattern:
\begin{align*}
        second \gets 
\begin{cases}
    second+1,& \text{with probability } 1/2\\
    0,              & \text{with probability } 1/2
\end{cases}
\end{align*}
When it reaches $S$, the $minute$ variable is incremented, and $second$ is reset back to 0. When $minute$ reaches $M$, the $hour$ variable is incremented and both other variables are reset to 0. Finally, the $hour$ and $minute$ variables are spread via the one-way epidemic process. By doing so, every agent learns the maximum $hour$ value in the system, and the maximum $minute$ value for its current $hour$.

The main innovation in our algorithm is the increment pattern that the $second$ variable undergoes. 
Our increment pattern guarantees that the $second$ variable is robust to adversarial interactions. What dictates the speed of the increment is the \emph{total number} of interactions in the system.

\RestyleAlgo{boxruled}
\LinesNumbered
\begin{algorithm}[htbp]
	\DontPrintSemicolon
	\caption{Phase clock}
	\label{alg: phase clock}
	$M\gets\Theta(\log n), S\gets\log (n/p) + \log \log n + O(1)$\\
	$\forall v\in V, v.second \gets0, v.minute \gets 0, v.hour \gets 0$\\
	
	\ForEach{interaction $(u,v)$}
    {
        $u$ makes a fair coin flip\\
        \If{Heads}{
            $u.second\gets u.second+1$\\
        }
        \Else{$u.second \gets 0$}
        \If{$u.second=S$}{
            $u.minute\gets u.minute+1$\\
            $u.second \gets 0$\\
        }
        \If{$u.minute=M$}{
            $u.hour\gets u.hour+1$\\
            $u.minute \gets 0$\\
        }
        //One-way epidemic \\
        \If{$u.hour < v.hour$}
        {
            $u.hour \gets v.hour$\\ 
            $u.minute \gets 0$\\
            $u.second \gets 0$\\
        }
        \If{$u.hour = v.hour$ and $u.minute < v.minute$}
        {
            $u.minute \gets v.minute$ \\
            $u.second \gets 0$\\
        }
    }
	
\end{algorithm}

In the following section, we show that indeed our algorithm is a phase clock with round length $\Theta(n p^{-1}\log^2 n)$ and stretch of $\Theta(n p^{-2}\log^2 n)$.

\subsection{Analysis}
\myparagraph{Lower bounding $L(i)$}
Our first goal is to show that $L(i)\geq\Omega(np^{-1}\log^2 n)$. In order to achieve this, it enough to show that $S(i)\geq\Omega(np^{-1}\log^2 n)$. This is due to the fact at as soon as a new maximum value for $hour$ appears in the population, it is spread to all agents via the one-way epidemic process within $O(n p^{-1} \log n)$ steps. Let us formalize this claim.
Assume that $S(i)=R_e(i)-R_e(i-1)\geq \Omega(n p^{-1} \log^2 n)$. On the other hand, due to the one-way epidemic it holds that $R_s(i) - R_e(i-1) \leq O(p^{-1} n \log n)$. Combining these two facts we get that:
$$L(i)=R_e(i) - R_s(i) \geq \Omega(n p^{-1}\log^2 n) +R_e(i-1) - R_s(i) \geq \Omega(n p^{-1} \log^2 n) - O(p^{-1}n \log n) =  \Omega(n p^{-1} \log^2 n)$$
Thus, for the rest of this section we focus on lower bounding $S(i)$.

As we aim to bound $S(i)=R_e(i)-R_e(i-1)$, for every $i$, let us assume for the rest of the analysis that $R_e(i-1)=0$. That is we assume that time 0 is when $v.hour=i$ holds for some agent for the first time.
Let $m' = \max_{v \in V, v.hour=i} v.minute$ and let $T_k$ be a random variable such that $m'=k$ holds in the $T_k$-th step for the first time. Note that $T_0 = 0$ holds with probability 1. We prove the following lemma:
\begin{lemma}
\label{lem: ub min}
For $c>2$, it holds that $Pr(T_{k+1} - T_k > c np^{-1} \log n) > 1/2$ for any $k=0,1,...,M-1$.
\end{lemma}
\begin{proof}
For $m'$ to increase by 1 starting from time $T_k$, at least one agent must observe $S$ consecutive heads in its coin flips.
Let us consider $cnp^{-1}\log n$ consecutive interactions starting from time 0. Let us denote by $x_v$ the amount of interactions agent $v$ took part in during this time as initiator. 
Note that $\sum_{v\in V} x_v = cnp^{-1}\log n$. 
Let us upper bound the probability of agent $v$ seeing $S$ consecutive heads during this time. The probability that agent $v$ sees a sequence of $S$ heads, starting exactly upon its  $j$-th interaction as initiator, and ending upon its $(\min \set{j+S, x_v})$-th interaction as initiator, is upper bounded by $2^{-S}$. Note that if $x_v-j<S$ the probability is 0, but the upper bound still holds. Now let us use a union bound over all values of $j$. This leads to an upper bound of $x_v\cdot2^{-S}$ for the probability that agent $i$ sees at least $S$ consecutive heads. Recall that $S=\log (n/p)+\log \log n+c$. To finish the proof we apply a union bound over all agents to get an upper bound of 
$$\sum_{v\in V} x_v\cdot2^{-S}= \frac{cnp^{-1} \log n}{2^c np^{-1} \log n} = c\cdot 2^{-c}$$
for the probability of at least one agent seeing $S$ consecutive heads over a period of $cnp^{-1}\log n$ interactions. Finally $Pr(T_{k+1} - T_k > c np^{-1} \log n) > 1-c\cdot 2^{-c}$. By setting $c>2$ we complete the proof.
\end{proof}

The above shows that with constant probability the maximum value of $minute$ among all agents does not increase too fast. This holds regardless of the value of $hour$. Next, we show that \emph{with high probability}, for every value of $i\geq 0$, the stretch of round $i$ is sufficiently large.


\begin{lemma}
\label{lem: hour lb}
For every $i\geq 0$, it holds with high probability that $L(i)=\Omega(np^{-1}\log^2 n)$.

\end{lemma}
\begin{proof}
Fix some $hour=i$, and let $X_k$ be the indicator variable for the event that $T_{k+1} - T_k > c n p^{-1} \log n$. 
According to Lemma~\ref{lem: ub min}, $E[X_k \mid X_0,...,X_{k-1}] > 1/2$ holds for any $k=0,1,...,M-1$. Let $X=\sum^{M-1}_{k=0} X_k$, and note that $S(i)\geq X\cdot c n p^{-1} \log n$. 
Ideally we would like to use a Chernoff bound to lower bound $X$, but unfortunately the $\set{X_k}$ variables are not independent. Thus, we apply Theorem~\ref{thm: modified chernoff} with parameters $q=1/2, t=M/4$ and get that:

\begin{align*}
\Pr\left[\sum^{M}_{i=0}X_i < M/4\right] 
\le  e^{-\Theta(M)}
\end{align*}

Recall that $M=\Theta(\log n)$, thus by setting $M$ sufficiently large we get that with high probability $S(i)=\Omega(np^{-1}\log^2 n)$. As noted before, this implies that $L(i)=\Omega(np^{-1}\log^2 n)$, which completes the proof.
\end{proof}

We note that the lower bound holds even when all interactions are chosen adversarially ($p=0$). The only reason $p$ appears in the lower bound is due to the definition of $S$. 
We continue to prove our upper bound, for which the existence of random interactions is crucial.
Specifically, we require the existence of random interactions in order to utilize one-way epidemics.

\myparagraph{Upper bounding $S(i)$} In what follows we upper bound $S(i)$ directly. As before, we first consider $m'$, the maximum value of the $minute$ variable, and show that it increases sufficiently fast.

\begin{lemma}
\label{lem:increase}
From any configuration where $m'=j < M$ holds, 
$m'$ increases to $j+1$ within $dn p^{-2} \log n$ steps with a constant probability, for a sufficiently large constant $d$.
\end{lemma}
\begin{proof}
Without loss of generality, we assume that
every agent satisfies $\hour=i$ and $\minute=j$ in the configuration because
the one way epidemic propagates $\max_{v \in V}v.\hour$
and $\max_{v \in V, v.hour=i} v.\minute$ to all agents within $O(p^{-1}n \log n)$ steps with high probability.

In our algorithm, we can say that
each agent $v \in V$ plays a \emph{lottery game} repeatedly. 
Agent $v$ starts one round of the game each time it sees a tail.
If $v$ sees $S$ consecutive heads before the next tail,
$v$ \emph{wins} the game in that round.
Otherwise, (i.e., if $v$ sees less than $S$ heads before it sees the next tail),
$v$ \emph{loses} the game in that round. 
When an agent sees the next tail (or wins the round), the next round of the game begins. 
At each round of the game, $v$ wins the game with probability $2^{-S} = \Omega(p/(n \log n))$
(Recall that $S=\log (np^{-1}) +\log \log n + O(1)$).
The goal of our proof is to show that with constant probability, 
some agent wins the game at least once within $dnp^{-2} \log n$ steps, for a sufficiently large constant $d$.

For an agent $v$, we denote by $W_i(v)$ the event that $v$ wins it's $i$-th game. Note that the $W_i(v)$ events are independent of each other for all values of $i$ and $v$. This is because for every interaction only the initiator flips a coin, and the coins used for every game don't overlap.
Let us also denote by $Y_k(v) = \bigvee^k_{i=1} W_i(v)$, the event that agent $v$ wins at least once in its first $k$ games. Let $Y_k = \bigvee_{v\in V} Y_k(v)$, be the event that at least one agent wins at least one of its first $k$ games.

Let $d$ be a sufficiently large constant and let $\tau = (d p^{-1} \log n)/4$. Let us denote by $X(v)$, the event that agent $v$ plays at least $\tau$ games in the first $d p^{-2} n\log n = 4np^{-1} \tau$ steps, and let $X = \bigwedge_{v \in V} X(v)$. Finally we are interested in lower bounding the probability of $Z$, the event that at least one agent wins a game within the first $4np^{-1} \tau$ steps. We note that it holds that $\Pr[Z] \geq \Pr[X \wedge Y_\tau]$. That is, if all agents play at least $\tau$ games within the first $4np^{-1} \tau$ steps, and at least one agents wins one of its first $\tau$ games then event $Z$ occurs. Applying a union bound, we write $\Pr[Z] = 1 - \Pr[\neg X \vee  \neg Y_\tau] \geq 1- \Pr[\neg X] - \Pr[ \neg Y_\tau]$. In order to conclude the proof we wish to show that $\Pr[ \neg Y_\tau] <1/3$ and $\Pr[\neg X] < 1/3$.

 To bound $\Pr [\neg X]$ it is sufficient to show that within the first $4np^{-1} \tau$ steps every agent sees at least $\tau$ tails with high probability. As every interaction is chosen uniformly at random with probability $p$, and the initiator / responder order is also random, within the first $4np^{-1} \tau$ steps each agent will observe at least $2 \tau$ tails
 in expectation. Applying a Chernoff bound for a sufficiently large constant $d$, we get that every agent will observe $\tau$ or more tails w.h.p. Thus $\Pr[\neg X] = O(1/n) < 1/3$.

Next, we bound $\Pr [\neg Y_\tau]$. Expanding the expression, we get:

\begin{align*}
  \Pr [\neg Y_\tau] = Pr\left[\bigwedge_{v \in V} \neg Y_\tau (v)\right] = Pr\left[\bigwedge_{v \in V} \bigwedge^\tau_{i=1}  \neg W_i(v)\right] = (1- 2^{-S})^{n\tau} \leq e^{-\frac{ n \tau}{2^S}} =e^{-\Omega(d)} \le 1/3,  
\end{align*}
where in the above we use the independence of the $\set{W_i(v)}$ events, and the fact that $d$ is sufficiently large. This completes the proof.

\end{proof}


We are now ready to prove our upper bound.
\begin{lemma}
\label{lem: hour ub}
For every $i\geq 0$, it holds with high probability that $S(i)=O(np^{-2}\log^2 n)$.
\end{lemma}
\begin{proof}
Let $d$ be a sufficiently large constant that satisfies the conditions of Lemma \ref{lem:increase}
and $\tau =  d n p^{-2} \log n$.
Fix some $hour=i$, and let $X_j$ be the indicator variable such that $X_j=1$ holds if and only if $m'$ increases at least by one
from the $(j-1)\tau$ step to the $\tau j-1$ step. 
By Lemma \ref{lem:increase}, it holds that $E[X_i \mid X_0,...,X_{i-1}] \geq q$ for some constant $q$. We finish the proof by applying Theorem~\ref{thm: modified chernoff} with parameter $t=M$ and get that:

\begin{align*}
\Pr\left[\sum^{ 2q^{-1} M}_{i=0}X_i < M\right] 
\le e^{-\Theta(M)}
\end{align*}
Recall that $M=\Theta(\log n)$, thus setting $M$ sufficiently large, implies that $S(i) \leq 2q^{-1} M\cdot d n p^{-2} \log n =O(np^{-2}\log^2 n)$ with high probability.
\end{proof}

Finally, we state our main theorem: 
\begin{theorem}
\label{thm: phase clock final}
Algorithm~\ref{alg: phase clock} is a $((np^{-1} \log^2 n), (np^{-2} \log^2 n))$-phase clock
that uses $O((\log n) \cdot (\log (p^{-1}n)))$ states for $\Gamma_s$.
%
\end{theorem}

\section{Leader election}
\label{sec: leader election}
We analyze the following leader election protocol, as described in \cite{SudoOIKM20} (the module backup)\footnotemark{}.\footnotetext{Essentially the same idea was presented in \cite{AAG18} before.
}
In the protocol every agent is either a \emph{a leader} or a \emph{follower}. This is represented via a binary $leader$ variable. We assume that initially there is at least one leader in the population. Every agent also has a $level$ variable initiated to 0 and bounded by the value $\lmax=\Omega(\log n)$. 
The protocol assumes the existence of a phase clock in the system. For every agent we call the time between two consecutive increases of the $hour$ variable of the phase clock an \emph{epoch} \emph{for that agent} (this is a subjective value per agent, not to be confused with a \emph{round} as was defined in the previous section). For our usage we can bound the range of the $hour$ variable by a small constant. This can be easily implemented via a modulo operation (see \cite{SudoOIKM20} for a detailed implementation), where the duration of each round still has the same guarantees of Theorem~\ref{thm: phase clock final}. To simplify the pseudo-code we introduce a $tick$ variable which is raised for an agent only in the first interaction it takes part in as an initiator once it enters a new epoch.

The algorithm consists of two parts, where the first part guarantees that we quickly converge to a single leader with high probability, while the second part guarantees that the population \emph{always} reaches a state where there exists a single leader. Accordingly, the second part is very slow to converge, but is rarely required.

\begin{enumerate}
    \item On the first interaction in each epoch, a leader makes a coin flip and increments the $level$ variable if it observed heads (up to the limit $\lmax$). Thereafter, the maximum level in the population is shared among all the agents via one-way epidemic. A leader becomes a follower when it observes a higher level than it's own.

    \item When two leaders interact, one remains a leader and the other one becomes a follower.

\end{enumerate}
 The pseudocode for the above is given in Algorithm~\ref{alg: leader election}.

\RestyleAlgo{boxruled}
\LinesNumbered
\begin{algorithm}[htbp]
	\DontPrintSemicolon
	\caption{Leader election}
	\label{alg: leader election}
	$\lmax\gets\Theta(\log n)$\\
	$\forall v\in V, v.level\gets 0$\\
	
	\ForEach{interaction $(u,v)$}
    {
        //One way epidemic\\
        \If{$u.level < v.level$}
        {
            $u.leader \gets false$\\
            $u.level \gets v.level$\\
        }
        //$u.tick\gets true$ when $u$ enters a new epoch\\ 
        \If{$u.tick = true$ and $u.leader = true$}
        {
        $u.tick \gets false$\\
        $u$ makes a fair coin flip\\
        \If{Heads}{
            $u.level\gets \min \set{u.level+1, \lmax}$\\
        }
        }
        
        \If{$v.leader = true$ and $u.leader = true$}{$u.leader \gets false$}
    }

\end{algorithm}


In \cite{SudoOIKM20} the correctness and running time are analyzed under $\Gamma_u$. First let us present the correctness analysis. That is, there is always at least one leader in the population. Roughly speaking, this is because the first part always keeps the leader with the highest level, while the second part only eliminates a leader if it interacts with another leader. So at any point in time when a leader is eliminated, it can "blame" a leader which currently exists.

As the run-time analysis of \cite{SudoOIKM20} is for $\Gamma_u$ it is not immediately clear what are the implications for $\Gamma_s$. Luckily, the analysis still goes through as long as we have a phase clock for $\Gamma_s$. Let us present a simplified analysis, which is somewhat different than the analysis presented in \cite{SudoOIKM20}. We aim to bound the time it takes to reduce the number of leaders to $1$. First we note that because the length of a round is $\Omega(p^{-1}n\log^2 n)$ steps, then with high probability every leader has at least one interaction during the first half of the round, also the information of the interaction is guaranteed to spread to the entire population within the round with high probability. This guarantees that for every round, every leader flips a coin, and the maximum level is propagated throughout the population via the one-way epidemic within that round.
For the rest of the analysis we assume that indeed every leader has at least one interaction per epoch and that the phase clock and one-way epidemic function correctly. As these events happen with high probability, we can guarantee that they hold throughout the execution of the first $\Theta(np^{-2} \log^3 n)$ steps with high probability via a simple union bound\footnote{Note that when the execution time goes to infinity, these guarantees (\ie synchronization via a phase clock) eventually fail. But our protocol has long since converged by this time.}.

Let us denote by $L_i$ the random variable for the number of leaders remaining after round $i$, where $L_0>0$ is the initial number of leaders. Then
it holds that $E[L_i]\leq L_{i-1}(1/2+2^{-L_{i-1}})$. This is because the distribution of $L_i$ behaves exactly like $B(L_{i-1}, 1/2)$ (number of heads when tossing $L_{i-1}$ fair coins), with the exception that if all coins are tails, we get $L_{i-1}$ leaders remaining instead of 0. Thus, when computing the expectation we must add a $L_{i-1}2^{-L_{i-1}}$ term. Finally, note that $L_{i-1}(1/2+2^{-L_{i-1}})\leq \frac{3}{4}L_{i-1}$ for all $L_{i-1}\geq 2$.
Using Markov's inequality, it holds that 
$$Pr[L_i \geq \frac{33}{40}L_{i-1}] = Pr[L_i \geq \frac{3}{4}L_{i-1}\cdot\frac{11}{10}] \leq \frac{10}{11}$$

Let us denote by $X_i$ the indicator random variable for the event that $L_i < \frac{33}{40}L_{i-1}$. Then it holds that $E[X_i \mid X_1,...,X_{i-1}]>q$ for some constant $q$.
To complete the proof we apply Theorem~\ref{thm: modified chernoff} and get that within $O(\log n)$ epochs we remain with a single leader with high probability. As every epoch requires $O(p^{-2}n\log^2 n )$ steps, we get a unique leader with high probability in $O(p^{-2}n\log^3 n )$ steps. 
Combining this with the cost of executing the slower second phase of the algorithm up to convergence, we get an expected running time of $O(p^{-2}n\log^3 n )$. This is because the second part requires $O(p^{-1}n^2)$ steps to completes, but is only required if the first part fails, which happens with probability $O(n^{-2})$. Thus, the second part's contribution to the expected stabilization time is $O(1/p)$.
We state the following theorem:
\begin{theorem}
\label{thm: leader election}
For every $p<1$, leader election can be solved under $\Gamma_s$ in $O(p^{-2}n\log^3 n)$ steps with high probability and in expectation using $\Theta((\log^2 n) \cdot (\log (np^{-1})))$ states. 
\end{theorem} 

\bibliography{paper6}

\clearpage

\appendix
\section{Random initiator-responder order}
We show how to extend our proof for the case where the initiator-responder order is random, and no random coins are available. First we change our phase clock algorithm to work without coin flips. The pseudo-code is given as Algorithm~\ref{alg: phase clock random order}. It is essentially the same algorithm, but now the initiator increases its $second$ value, and the responder sets it to 0. Throughout the rest of the analysis we still refer to "coin flips" made by the agents, where we mean that an agent flips heads if it is an initiator, and tails otherwise. The most important difference to keep in mind is that, while the coin flips for each agent are independent of each other, coin flips between \emph{different agents} are no longer independent. 

\RestyleAlgo{boxruled}
\LinesNumbered
\begin{algorithm}[htbp]
	\DontPrintSemicolon
	\caption{Phase clock}
	\label{alg: phase clock random order}
	$M\gets\Theta(\log n), S\gets\log (n/p) + \log \log n + O(1)$\\
	$\forall v\in V, v.second \gets0, v.minute \gets 0, v.hour \gets 0$\\
	
	\ForEach{interaction $(u,v)$}
    {
        
        $u.second\gets u.second+1$\\
        $v.second \gets 0$\\
        \If{$u.second=S$}{
            $u.minute\gets u.minute+1$\\
            $u.second \gets 0$\\
        }
        \If{$u.minute=M$}{
            $u.hour\gets u.hour+1$\\
            $u.minute \gets 0$\\
        }
        //One-way epidemic \\
        \If{$u.hour < v.hour$}
        {
            $u.hour \gets v.hour$\\ 
            $u.minute \gets 0$\\
            $u.second \gets 0$\\
        }
        \If{$u.hour = v.hour$ and $u.minute < v.minute$}
        {
            $u.minute \gets v.minute$ \\
            $u.second \gets 0$\\
        }
    }
	
\end{algorithm}

We now restate the relevant Lemmas. First note that the lower bound on $L(i)$ still holds as we did not assume independence between coin flips made by different agents in the proof of Lemma~\ref{lem: hour lb} and Lemma~\ref{lem: ub min}. As for the upper bound we did assume independence in Lemma~\ref{lem:increase}, but not in Lemma~\ref{lem: hour ub}.
We state the following alternative for Lemma~\ref{lem:increase}:
\begin{lemma}
\label{lem:increase alt}
For Algorithm~\ref{alg: phase clock random order}, from any configuration where $m'=j < M$ holds, 
$m'$ increases to $j+1$ within $O(n p^{-2} \log^2 n)$ steps with a constant probability.
\end{lemma}
\begin{proof}

We maintain the same notations as the proof of Lemma~\ref{lem:increase}, with the exception that we choose $\tau = (S dp^{-1} \log n)/4$ (larger by an $S$ factor than originally). The proof remains unchanged until we need to bound $Pr[\neg X]$ and $Pr[\neg Y]$. Now there exist dependencies between the coin flips of different agents (but not the coin flips of a single agent).

The proof that $Pr[\neg X]$, remains unchanged. That is, we used a Chernoff bound to state that $Pr[\neg X(v)] < 1/n^2$. Now we have dependencies between coin flips of different agents, however the coin flips of a single agent are still independent. Finally, we note that:
$$Pr[\neg X] = Pr\left[\bigvee_{v\in V} \neg X(v)\right] \leq 1/n$$
Where the last transition is due to a union bound.

Next we bound $Pr[\neg Y]$. Again, we expand the expression:

\begin{align*}
  \Pr [\neg Y_\tau] = Pr\left[\bigwedge_{v \in V} \neg Y_\tau (v)\right] = Pr\left[\bigwedge_{v \in V} \bigwedge^\tau_{i=1}  \neg W_i(v)\right] \leq  Pr\left[\bigwedge_{(v,i) \in U}  \neg W_i(v)\right]
\end{align*}

Our goal is to bound $Pr\left[\bigwedge_{v \in V} \neg Y_\tau (v)\right]$, however there are dependencies between the events. However, it is sufficient if we can find a subset $U \subseteq V \times [\tau]$, such that the events $\set{W_i(v)}_{(v,i)\in U }$ are independent. Note that that every $W_i(v)$ can depend on at most $S$ other events. This is because every round can have length at most $S$ (at which point the round is won). Let us now construct a set $U$ corresponding to \emph{independent} events $\set{W_i(u)}_{(u,i)\in V}$. This can be constructed greedily, starting with $U=\emptyset, V'=V\times[\tau]$, we add some $(u,i)\in V'$ to $U$ and remove from $V'$ all $(v,j)$ such that event $W_i(u)$ depends on $W_j(v)$. We continue this construction until $V'=\emptyset$. As for every element added to $U$ at most $S$ elements were removed from $V'$, we get that $\size{U} \geq n\tau / S$. Now we can write:

\begin{align*}
  Pr\left[\bigwedge_{(v,i) \in U}  \neg W_i(v)\right] \leq (1- 2^{-S})^{n\tau / S} \leq e^{-\frac{ (n dp^{-1} \log n)/4 }{2^S}} = e^{-\frac{ (n dp^{-1} \log n)/4 }{2^c n p^{-1}\log n}}  =e^{-\Omega(d)} \le 1/3, 
\end{align*}

which completes the proof.

\end{proof}
We now state the main theorem for our phase clock.
\begin{theorem}
\label{thm: phase clock final alt}
Algorithm~\ref{alg: phase clock} is a $(np^{-1} \log^2 n,~np^{-2} \log^3 n)$-phase clock
that uses $O((\log n) \cdot (\log (p^{-1}n)))$ states
under $\Gamma_s$ when the initiator-responder order is random.
%
%
\end{theorem}

Finally let us restate our leader election algorithm for the random initiator-responder order case (Algorithm~\ref{alg: leader election alt}). As before, we can still see this algorithm as flipping coins, but we lose the independence. A slight detail we must notice is that according to our original definition the tick is raised when the agent enters a new epoch. This is now problematic, as when an agent enters a new epoch it is always an initiator, and thus will always increase its level. To overcome this obstacle we can assume that tick is raised one interaction \emph{after} the $hour$ variable was increased. This now gives an equal probability for increasing and not increasing the level variable.
Our analysis presented in Section \ref{sec: leader election} does not require independence between coin flips and it goes through unchanged. Thus, we have the following theorem.

\RestyleAlgo{boxruled}
\LinesNumbered
\begin{algorithm}[htbp]
	\DontPrintSemicolon
	\caption{Leader election}
	\label{alg: leader election alt}
	$\lmax\gets\Theta(\log n)$\\
	$\forall v\in V, v.level\gets 0$\\
	
	\ForEach{interaction $(u,v)$}
    {
        //One way epidemic\\
        \If{$u.level < v.level$}
        {
            $u.leader \gets false$\\
            $u.level \gets v.level$\\
        }
        //$u.tick\gets true$ one interaction after $u$ enters a new epoch\\ 
        
        \If{$u.tick = true$ and $u.leader = true$}
        {
        $u.tick \gets false$\\
        $u.level\gets \min \set{u.level+1, \lmax}$\\
        }
        \If{$v.tick=true$}
        {
        $u.tick \gets false$\\
        }
        
        \If{$v.leader = true$ and $u.leader = true$}{$v.leader \gets false$}
        
    }

\end{algorithm}

\begin{theorem}
For every $p<1$, Algorithm~\ref{alg: leader election alt} solves leader election under $\Gamma_s$ in $O(np^{-2}\log^4 n)$ steps with high probability and in expectation, using $\Theta((\log^2 n) \cdot (\log (n/p)))$ states, when the initiator-responder order is random.
\end{theorem} 

\end{document}